\documentclass[10pt]{article}
\usepackage[T1]{fontenc}
\usepackage[latin9]{inputenc}
\usepackage[b5paper]{geometry}
\usepackage{color}
\usepackage{amsmath}
\usepackage{amsthm}
\usepackage{amssymb}
\usepackage[numbers]{natbib}
\usepackage[unicode=true,pdfusetitle,
 bookmarks=true,bookmarksnumbered=false,bookmarksopen=false,
 breaklinks=false,pdfborder={0 0 1},backref=false,colorlinks=true]
 {hyperref}

\makeatletter
\newenvironment{lyxlist}[1]
	{\begin{list}{}
		{\settowidth{\labelwidth}{#1}
		 \setlength{\leftmargin}{\labelwidth}
		 \addtolength{\leftmargin}{\labelsep}
		 }}
	{\end{list}}
\theoremstyle{plain}
\newtheorem{thm}{\protect\theoremname}
\theoremstyle{plain}
\newtheorem{prop}[thm]{\protect\propositionname}
\theoremstyle{definition}
\newtheorem{defn}[thm]{\protect\definitionname}
\theoremstyle{remark}
\newtheorem{rem}[thm]{\protect\remarkname}
\theoremstyle{plain}
\newtheorem{lem}[thm]{\protect\lemmaname}
\theoremstyle{plain}
\newtheorem{cor}[thm]{\protect\corollaryname}

\usepackage{latexsym}
\usepackage{enumitem}
\setenumerate{itemsep=2pt,partopsep=0pt,parsep=0pt,topsep=3pt}
\setitemize{itemsep=2pt,partopsep=0pt,parsep=0pt,topsep=3pt}
\setdescription{itemsep=2pt,partopsep=0pt,parsep=0pt,topsep=3pt}

\date{}

\makeatother

\providecommand{\corollaryname}{Corollary}
\providecommand{\definitionname}{Definition}
\providecommand{\lemmaname}{Lemma}
\providecommand{\propositionname}{Proposition}
\providecommand{\remarkname}{Remark}
\providecommand{\theoremname}{Theorem}

\begin{document}

\title{Some Common Mistakes in the Teaching and Textbooks of Modal Logic}
\author{Xuefeng Wen\\
Institute of Logic and Cognition\\
Department of Philosophy\\
Sun Yat-sen University\\
wxflogic@gmail.com}
\maketitle
\begin{abstract}
We discuss four common mistakes in the teaching and textbooks of modal
logic. The first one is missing the axiom $\Diamond\varphi\leftrightarrow\neg\Box\neg\varphi$,
when choosing $\Diamond$ as the primitive modal operator, misunderstanding
that $\Box$ and $\Diamond$ are symmetric. The second one is forgetting
to make the set of formulas for filtration closed under subformulas,
when proving the finite model property through filtration, neglecting
that $\Box\varphi$ and $\Diamond\varphi$ may be abbreviations of
formulas. The third one is giving wrong definitions of canonical relations
in minimal canonical models that are unmatched with the primitive
modal operators. The final one is misunderstanding the rule of necessitation,
without knowing its distinction from the rule of modus ponens. To
better understand the rule of necessitation, we summarize six ways
of defining deductive consequence in modal logic: omitted definition,
classical definition, ternary definition, reduced definition, bounded
definition, and deflationary definition, and show that the last three
definitions are equivalent to each other.

\end{abstract}

\section{Primitive Modal Operators and the d/Dual Axiom}

In modal logic, we can set $\Box$ as the primitive modal operator,
and define $\Diamond$ to be $\neg\Box\neg$ as a derived operator,
as in \citep{Chagrov1997}. We can also set $\Diamond$ as the primitive
operator, defining $\Box$ to be $\neg\Diamond\neg$, as in \citep{Blackburn2001a}.
Though the two options seem totally symmetric, they are actually not.
In constructing axiomatic systems, the choice of different primitive
modal operators may lead to results that are not totally symmetric.

For example, consider the minimal normal modal logic $\mathbf{K}$.
When $\Box$ is the primitive modal operator, apart from the axioms
and rules for classical propositional logic ($\mathrm{PC}$ henceforth),
the axiomatic system needs only include the following axiom (schema)
and rule of inference,
\begin{lyxlist}{00.00.0000}
\item [{$\mathrm{K}$}] $\Box(\varphi\to\psi)\to(\Box\varphi\to\Box\psi)$
\item [{$\mathrm{RN}$}] $\dfrac{\varphi}{\Box\varphi}$
\end{lyxlist}
as well as the definition $\Diamond\varphi=_{df}\neg\Box\neg\varphi$,
which amounts to adding the following axiom.
\begin{lyxlist}{00.00.0000}
\item [{$\mathrm{dual}$}] $\Diamond\varphi\leftrightarrow\neg\Box\neg\varphi$
\end{lyxlist}
But when $\Diamond$ is the primitive modal operator, apart from $\mathrm{PC},\mathrm{K},\mathrm{RN}$
and the following axiom used as a definition,
\begin{lyxlist}{00.00.0000}
\item [{$\mathrm{Dual}$}] $\Box\varphi\leftrightarrow\neg\Diamond\neg\varphi$
\end{lyxlist}
we still need to augment $\mathrm{dual}$. From another perspective,
if we do not take definitions as a part of the axiomatic system, when
choosing $\Box$ as primitive, apart from $\mathrm{PC}$ and $\mathrm{RN}$,
we need only one more axiom, namely $\mathrm{K}$; when choosing $\Diamond$
as primitive, however, we need two axioms, namely $\mathrm{K}$ and
$\mathrm{dual}$. Why is there such an asymmetry? The reason is that,
the axiom $\mathrm{K}$ and the rule $\mathrm{RN}$ are intrinsically
using $\Box$ rather than $\Diamond$. This makes the following rule
derivable from $\mathrm{K}$, $\mathrm{RN}$, and $\mathrm{dual}$
(plus $\mathrm{PC}$).
\begin{lyxlist}{00.00.0000}
\item [{$\mathrm{RE}$}] $\dfrac{\varphi\leftrightarrow\psi}{\Box\varphi\leftrightarrow\Box\psi}$
\end{lyxlist}
The derivation is as follows.
\begin{enumerate}
\item $\varphi\leftrightarrow\psi$\hfill hypothesis
\item $\varphi\to\psi$, $\psi\to\varphi$\hfill (1), PC
\item $\Box(\varphi\to\psi)$, $\Box(\psi\to\varphi)$\hfill (2), RN
\item $\Box(\varphi\to\psi)\to(\Box\varphi\to\Box\psi)$, $\Box(\psi\to\varphi)\to(\Box\psi\to\Box\varphi)$\hfill
K
\item $\Box\varphi\to\Box\psi$, $\Box\psi\to\Box\varphi$\hfill (3), (4),
PC
\item $\Box\varphi\leftrightarrow\Box\psi$\hfill (5), PC
\end{enumerate}
With $\mathrm{RE}$, $\mathrm{Dual}$ can be derived from $\mathrm{dual}$,
as shown below.
\begin{enumerate}
\item $\Diamond\neg\varphi\leftrightarrow\neg\Box\neg\neg\varphi$\hfill
dual
\item $\neg\Diamond\neg\varphi\leftrightarrow\Box\neg\neg\varphi$\hfill
(1), PC
\item $\varphi\leftrightarrow\neg\neg\varphi$\hfill PC
\item $\Box\varphi\leftrightarrow\Box\neg\neg\varphi$\hfill (3), RE
\item $\Box\varphi\leftrightarrow\neg\Diamond\neg\varphi$\hfill (2), (4),
PC
\end{enumerate}
Hence, when $\Box$ is primitive, there is no need to add $\mathrm{Dual}$
when $\mathrm{dual}$ is available. Dually, if we want to derive $\mathrm{dual}$
from $\mathrm{Dual}$, we need the following rule.
\begin{lyxlist}{00.00.0000}
\item [{$\mathrm{re}$}] $\dfrac{\varphi\leftrightarrow\psi}{\Diamond\varphi\leftrightarrow\Diamond\psi}$
\end{lyxlist}
With $\mathrm{re}$, we can analogously derive $\mathrm{dual}$ from
$\mathrm{Dual}$ as follows.
\begin{enumerate}
\item $\Box\neg\varphi\leftrightarrow\neg\Diamond\neg\neg\varphi$\hfill
Dual
\item $\neg\Box\neg\varphi\leftrightarrow\Diamond\neg\neg\varphi$\hfill
(1), PC
\item $\varphi\leftrightarrow\neg\neg\varphi$\hfill PC
\item $\Diamond\varphi\leftrightarrow\Diamond\neg\neg\varphi$\hfill (3),
re
\item $\Diamond\varphi\leftrightarrow\neg\Box\neg\varphi$\hfill (2), (4),
PC
\end{enumerate}
The problem, however, is that with only $\mathrm{K}$, $\mathrm{RN}$,
and $\mathrm{PC}$, even with $\mathrm{Dual}$, we can not obtain
$\mathrm{re}$. With $\mathrm{RE}$ (which can be obtained from $\mathrm{K}$,
$\mathrm{RN}$, and $\mathrm{PC}$) and $\mathrm{Dual}$, we can only
obtain
\[
\dfrac{\varphi\leftrightarrow\psi}{\Diamond\neg\varphi\leftrightarrow\Diamond\neg\psi},\dfrac{\varphi\leftrightarrow\psi}{\Diamond\neg\neg\varphi\leftrightarrow\Diamond\neg\neg\psi},\dfrac{\varphi\leftrightarrow\psi}{\Diamond\neg\neg\neg\varphi\leftrightarrow\Diamond\neg\neg\neg\psi}\cdots\cdots
\]
but not $\mathrm{re}$. Thus we can not derive $\mathrm{dual}$ from
$\mathrm{Dual}$. Thereby when $\Diamond$ is primitive, apart from
$\mathrm{Dual}$ used as a definition, we have to add $\mathrm{dual}$
too. Since we have proved that $\mathrm{Dual}$ can be derived from
$\mathrm{dual}$, the axiom $\mathrm{Dual}$ used as a definition
can be omitted. This implies that whichever of $\Box$ and $\Diamond$
is chosen as primitive, $\mathrm{dual}$ is indispensable (either
as a definition, or as an axiom), while $\mathrm{Dual}$ can be omitted.

Let $\tilde{\mathbf{K}}$ be the axiomatic system consisting of $\mathrm{PC}$,
$\mathrm{K}$, $\mathrm{RN}$, and $\mathrm{Dual}$. We will prove
rigorously that $\mathrm{dual}$ is not derivable in $\tilde{\mathbf{K}}$.
\begin{prop}
\label{prop:non-derivable}$\nvdash_{\tilde{\mathbf{K}}}\Diamond p\leftrightarrow\neg\Box\neg p$
\end{prop}

\begin{proof}
Consider $\mathfrak{M}=(W,R,V)$ and the following non-standard semantics.
\begin{itemize}
\item $\mathfrak{M},w\models\Diamond\varphi$ iff $\varphi=\neg\psi$ and
there exists $u\in W$ such that $wRu$ and $\mathfrak{M},u\not\models\psi$;
\item $\mathfrak{M},w\models\Box\varphi$ iff for all $u\in W$, $wRu$
implies $\mathfrak{M},u\models\varphi$.
\end{itemize}
The semantics of propositional connectives is defined as usual. It
is easily verified that $\mathrm{PC}$ and $\mathrm{K}$ are valid
under this semantics (with respect to the class of all frames), and
$\mathrm{MP}$ and $\mathrm{RN}$ preserves validity. Notice that
$\mathfrak{M},w\models\Diamond\neg\varphi$ iff there exists $u\in W$
such that $wRu$ and $\mathfrak{M},u\not\models\varphi$. Thus $\mathrm{Dual}$,
namely $\Box\varphi\leftrightarrow\neg\Diamond\neg\varphi$ is also
valid under this semantics. Hence, if $\vdash_{\tilde{\mathbf{K}}}\Diamond p\leftrightarrow\neg\Box\neg p$,
then $\Diamond p\leftrightarrow\neg\Box\neg p$ is valid too. But
consider the counter-model $\mathfrak{M}=(\{w\},\{(w,w)\},V)$, where
$V(p)=\{w\}$. Then $\mathfrak{M},w\models\neg\Box\neg p$ but $\mathfrak{M},w\not\models\Diamond p$.
Therefore, $\nvdash_{\tilde{\mathbf{K}}}\Diamond p\leftrightarrow\neg\Box\neg p$.
\end{proof}
Early textbooks in modal logic (such as \citep{Segerberg1971,Lemmon1977,Hughes1996})
usually take $\Box$ to be primitive, with $\mathrm{dual}$ a definition
rather than an axiom added to the axiomatic systems. Moreover, $\mathrm{Dual}$
is not required as an axiom. The prevalence of \citep{Blackburn2001a}
makes more and more people choose $\Diamond$ to be primitive. They
may take for granted from duality that only $\mathrm{Dual}$ should
be added as a definition and $\mathrm{dual}$ is not required as an
axiom. The above analysis shows that this thought is incorrect. If
$\Diamond$ is taken as primitive, the axiomatic system with only
$\mathrm{Dual}$ and without $\mathrm{dual}$ is incomplete. A newly
published textbook \citep{Pacuit2017} on neighborhood semantics just
made this mistake. On page 54, the author claims that the minimal
modal logic $\mathbf{E}$ under neighborhood semantics can be axiomatized
by $\mathrm{PC}$, $\mathrm{RE}$, and $\mathrm{Dual}$. But a slight
modification of the proof of Proposition~\ref{prop:non-derivable}
will show that this axiomatic system is incomplete, as $\mathrm{dual}$
is not derivable in the system. The correct axiomatization is to take
$\mathrm{dual}$ instead of $\mathrm{Dual}$ as an axiom.

In semantics, the choice of primitive modal operators will affect
the definition of filtration and minimal canonical models, as well
as the syntax of subformulas. Without caution, some subtle mistakes
are likely to be made, which will be shown in Section 2 and 3.

Given a set of propositional variables $PV$, without other specification,
we assume the language of modal logic is defined as follows.

\[
\mathcal{L}_{\Diamond}\ni\varphi::=p\mid\neg\varphi\mid(\varphi\to\varphi)\mid\Diamond\varphi,
\]
where $p\in PV$. The other logical connectives ($\top,\lor,\land,\leftrightarrow$)
are defined as usual.

\section{Filtration and Finite Model Property}

The basic idea of filtration is as follows. Given a formula $\varphi$
and its counter-model $\mathfrak{M}$, the satisfiability of $\varphi$
only depends on the satisfiability of the subformulas of $\varphi$.
Since the subformulas are finite, the possibilities of the satisfiability
of them are also finite. Thus, if we take those points in $\mathfrak{M}$
that satisfy the same subformulas of $\varphi$ to be the same, we
obtain a finite model $\mathfrak{M}^{f}$. If we define in $\mathfrak{M}^{f}$
an accessibility relation that is closely related to $\mathfrak{M}$
such that the satisfiability of the subformulas in $\mathfrak{M}^{f}$
is equivalent to that in $\mathfrak{M}$ for all subformulas of $\varphi$,
then we obtain a finite counter-model of $\varphi$.

Given a model $\mathfrak{M}=(W,R,V)$ and a set of formulas $\Sigma\subseteq\mathcal{L}_{\Diamond}$,
we first define an equivalence relation $\sim_{\Sigma}\ \subseteq W\times W$
as follows.
\[
w\sim_{\Sigma}w'\mbox{ iff for all }\varphi\in\Sigma,\mathfrak{M},w\models\varphi\Leftrightarrow\mathfrak{M},w'\models\varphi.
\]
Define the equivalence class $|w|_{\Sigma}$ of $w$ as follows.\footnote{We often omit the subscript $\Sigma$ if no confusion occurs.}
\[
|w|_{\Sigma}=\{w'\in W\mid w\sim_{\Sigma}w'\}
\]
When $\Sigma$ is finite, the set of equivalence classes induced by
$\Sigma$ is also finite. When $\Sigma$ is used for filtration, it
is required to be subformula closed, i.e., every subformula of every
formula in $\Sigma$ is also in $\Sigma$. In the sequel, we denote
by $Sub\Sigma$ the set of all subformulas of all formulas in $\Sigma$.
If $\Sigma$ is a singleton $\{\varphi\}$, we denote $Sub\Sigma$
by $Sub\varphi$. Now recall the definition of filtration.
\begin{defn}
[filtration]\label{def:filtration}$\mathfrak{M}^{f}=(W^{f},R^{f},V^{f})$
is a filtration of $\mathfrak{M}=(W,R,V)$ through $\Sigma$, if the
following conditions are satisfied.
\begin{enumerate}
\item $W^{f}=\{|w|_{\Sigma}\mid w\in W\}$;
\item for all $w,u\in W$, if $wRu$ then $|w|R^{f}|u|$;
\item for all $w,u\in W$, if $|w|R^{f}|u|$, then for all $\Diamond\varphi\in\Sigma$,
if $\mathfrak{M},u\models\varphi$ then $\mathfrak{M},w\models\Diamond\varphi$;
\item for all $w\in W$ and $p\in PV\cap\Sigma$, $|w|\in V^{f}(p)$ iff
$w\in V(p)$.
\end{enumerate}
\noindent{}Such an $R^{f}$ is also called a filtration of $R$ (through
$\Sigma$).
\end{defn}

A few remarks about filtration.
\begin{rem}
\label{rem:filtration}If $\Sigma$ is subformula closed and the primitive
modal operator is $\Diamond$, then $(3)$ above implies 
\begin{enumerate}
\item [$(3')$] For all $w,u\in W$, if $|w|R^{f}|u|$, then for all $\Box\varphi\in\Sigma$,
if $\mathfrak{M},w\models\Box\varphi$ then $\mathfrak{M},u\models\varphi$.
\end{enumerate}
But $(3')$ does not imply $(3)$. The reason is that $\Box\varphi$
is actually an abbreviation of $\neg\Diamond\neg\varphi$. When $\Box\varphi\in\Sigma$,
by the subformula closure of $\Sigma$, we also have $\neg\varphi\in\Sigma$,
and thus we could prove $(3')$ using $(3)$ by contraposition. But
conversely, when $\Diamond\varphi\in\Sigma$, we can not obtain $\neg\varphi\in\Sigma$,
and thus can not prove $(3)$ using $(3')$ by contraposition. On
this point, filtrations are different from canonical models and ultrafilter
extensions. In canonical models and ultrafilter extensions, the definitions
using $\Box$ can dually be replaced by $\Diamond$, and vice versa.
More precisely, the canonical relation $R^{\Lambda}$ for the logic
$\Lambda$ can be defined by: $R^{\Lambda}wu$ iff for all $\varphi\in u$,
$\Diamond\varphi\in w$. It can also be defined by: $R^{\Lambda}wu$
iff for all $\Box\varphi\in w$, $\varphi\in u$. The two definitions
are completely equivalent (assuming $\Lambda$ contains $\mathrm{dual}$).
Analogously, the accessibility relation $R^{ue}$ of the ultrafilter
extension of $\mathfrak{M}=(W,R,V)$ can be defined by: $R^{ue}wu$
iff for all $X\in v$, $\Diamond_{R}X\in u$.\footnote{$\Diamond_{R}X=\{w\in W\mid\exists x\in X\ Rwx\}$, $\Box_{R}X=\{w\in W\mid\forall x\in W(Rwx\Rightarrow x\in X)\}$.}
It can also be defined by: $R^{ue}wu$ iff for all $\Box_{R}X\in u$,
$X\in v$. The two definitions are also equivalent. But for filtrations,
(3) and ($\mathrm{3}'$) are not equivalent.
\end{rem}

\begin{rem}
If the primitive modal operator is $\Box$, for the inductive proof
of the case $\Box\varphi$ for the filtration theorem below, $(3)$
in Definition~\ref{def:filtration} should be replaced by $(3')$.
Then $(3')$ implies $(3)$ but $(3)$ does not imply $(3')$. The
reasons are as above.
\end{rem}

\begin{rem}
If both $\Box$ and $\Diamond$ are primitive modal operators, for
the inductive proof of the cases $\Box\varphi$ and $\Diamond\varphi$
for the filtration theorem below, $(3)$ in Definition~\ref{def:filtration}
should be replaced by $(3'')$.
\begin{enumerate}
\item [$(3'')$] For all $w,u\in W$, if $|w|R^{f}|u|$, then for all
$\Diamond\varphi\in\Sigma$, if $\mathfrak{M},u\models\varphi$ then
$\mathfrak{M},w\models\Diamond\varphi$, and for all $\Box\varphi\in\Sigma$,
if $\mathfrak{M},w\models\Box\varphi$ then $\mathfrak{M},u\models\varphi$.
\end{enumerate}
i.e., the conjunction of $(3)$ and $(3'')$.
\end{rem}

Note that the definition of $R^{f}$ in Definition~\ref{def:filtration}
is not constructive. The following two particular $R^{f}$s are often
used.
\begin{defn}
[smallest filtration, largest filtration]\label{def:minimal-maximal-filtration}Given
a model $\mathfrak{M}=(W,R,V)$ and a set of formulas $\Sigma\subseteq\mathcal{L}_{\Diamond}$,
the smallest filtration $R^{s}$ and the largest filtration $R^{l}$
of $R$ through $\Sigma$ is defined respectively as follows. For
all $w,u\in W$,
\begin{enumerate}
\item $|w|R^{s}|u|$ iff there exist $w'\in|w|$ and $u'\in|u|$ such that
$w'Ru'$;
\item $|w|R^{l}|u|$ iff for all $\Diamond\varphi\in\Sigma$, if $\mathfrak{M},u\models\varphi$
then $\mathfrak{M},w\models\Diamond\varphi$.
\end{enumerate}
Similarly, when the primitive modal operator is $\Box$, $(2)$ should
be replaced by $(2')$ below.
\begin{enumerate}
\item [(2')]$|w|R^{l}|u|$ iff for all $\Box\varphi\in\Sigma$, if $\mathfrak{M},w\models\Box\varphi$
then $\mathfrak{M},u\models\varphi$.
\end{enumerate}
If both $\Box$ and $\Diamond$ are primitive, then $(2)$ should
be replaced by $(2'')$ below.
\begin{enumerate}
\item [$(2'')$]$|w|R^{l}|u|$ iff for all $\Diamond\varphi\in\Sigma$,
if $\mathfrak{M},u\models\varphi$ then $\mathfrak{M},w\models\Diamond\varphi$
and for all $\Box\varphi\in\Sigma$, if $\mathfrak{M},w\models\Box\varphi$
then $\mathfrak{M},u\models\varphi$.
\end{enumerate}
\end{defn}

The equivalence of satisfiability of related formulas between a model
and its filtration is ensured by the following filtration theorem.
\begin{thm}
[Filtration Theorem]Let $\mathfrak{M}^{f}=(W^{f},R^{f},V^{f})$
be a filtration of $\mathfrak{M}=(W,R,V)$ through $\Sigma$, which
is subformula closed. Then for all $w\in W$ and $\varphi\in\Sigma$,
\[
\mathfrak{M},w\models\varphi\text{ iff }\mathfrak{M}^{f},|w|\models\varphi.
\]
\end{thm}

It should be emphasized that the subformula closed property of $\Sigma$
is critical in the proof of the filtration theorem; otherwise, we
could not use the inductive hypothesis to complete the proof.

Sometimes using the subformulas of the formula in question directly
will not achieve our goal. We have to supplement some other formulas
for filtration. It should be noticed that after adding more formulas
we still have to keep the set of formulas to be subformula closed.
At this place, both \citep{Chagrov1997} and \citep{Popkorn1994}
made an incautious mistake in the proof of the finite model property
of $\mathbf{K5}$.

The proof in \citep[p. 145]{Chagrov1997} goes as follows.
\begin{quotation}
First, $\mathbf{K5}=\mathbf{K}\oplus\Diamond\Box p\to\Box p$ is characterized
by the class of Euclidean frames. Let $\mathfrak{M}$ be a countermodel
for a formula $\varphi$ based on a Euclidean frame. Again, a filtration
of $\mathfrak{M}$ through $Sub\varphi$ need not be Euclidean. So
let us try a bigger filter, say,
\[
\Sigma=Sub\varphi\cup\{\Diamond\Box\varphi\mid\Box\psi\in Sub\varphi\}.
\]
Let $\mathfrak{N}$ be the largest filtration of $\mathfrak{M}$ through
$\Sigma$. We show that its underlying frame $\mathfrak{G}=(V,S)$
is Euclidean.

Suppose $|x|S|y|$ and $|x|S|z|$, for some $|x|,|y|,|z|\in V$, and
prove that $|y|S|z|$. By the definition of $S$, we need to show
that $\mathfrak{N},|y|\models\Box\psi$ implies $\mathfrak{N},|z|\models\psi$,
for every $\Box\psi\in\Sigma$. So let $\Box\psi\in\Sigma$ and $\mathfrak{N},|y|\models\Box\psi$.
Then $\mathfrak{N},|x|\models\Diamond\Box\psi$ and, by the filtration
theorem, $\mathfrak{M},x\models\Diamond\Box\psi$, from which $\mathfrak{M},x\models\Box\psi$,
since $\mathfrak{M}$ is a model for $\mathbf{K5}$. Therefore, $\mathfrak{N},|x|\models\Box\psi$
and $\mathfrak{N},|z|\models\psi$.
\end{quotation}
Notice that since \citep{Chagrov1997} takes $\Box$ as primitive,
the largest filtration $R^{f}$ is defined by $(2')$ in Definition~\ref{def:minimal-maximal-filtration}
instead of $(2)$. In the proof, the conclusion $\mathfrak{N},|x|\models\Diamond\Box\psi$
is obtained from $\Diamond\Box\psi\in\Sigma$ and the claim in Remark~\ref{rem:filtration}:
when $\Box$ is primitive, $(3')$ implies $(3)$. The whole proof
seems innocuous and is much shorter than the standard proof using
finite bases\footnote{See Definition~\ref{def:finite-base}.}. Unfortunately,
the proof is incorrect!

What is incorrect is that $\Sigma$ is not subformula closed as it
appears. The author thought wrongly that $Sub\Diamond\Box\psi=\{\Diamond\Box\psi,\Box\psi\}\cup Sub\psi$.
If this is case, then the proof does go through. The problem is that
$\Diamond$ is not a primitive operator, but the abbreviation of $\neg\Box\neg$.
Hence,
\[
Sub\Diamond\Box\psi=Sub\neg\Box\neg\Box\psi=\{\neg\Box\neg\Box\psi,\Box\neg\Box\psi,\neg\Box\psi,\Box\psi\}\cup Sub\psi.
\]
But $\Box\neg\Box\psi$ and $\neg\Box\psi$ are not in $\Sigma$.
Can the problem be solved by supplementing these two formulas? No,
because the proof requires that if $\Box\neg\Box\psi\in\Sigma$ then
$\Diamond\Box\neg\Box\psi$ must also be in $\Sigma$. But then it
will introduce more formulas of the form $\Box\psi$, and the proof
requires the formulas with $\Diamond$ prefixed to these formulas
must also be in $\Sigma$. Repeating this process, $\Sigma$ will
become an infinite set. To complete the proof, we have to prove that
$\Sigma$ has finite base with respect to $\mathfrak{M}$. Can we
solve the problem by just taking both $\Box$ and $\Diamond$ to be
primitive? No. Though $\Sigma$ is now subformula closed, the definition
for $R^{l}$ has to be revised as $(2'')$ in Definition~\ref{def:minimal-maximal-filtration}
. Then to prove that $\mathfrak{N}$ is Euclidean we have to consider
both $(2)$ and $(2'')$, which requires $\Sigma$ to satisfy not
only that $\Box\psi\in\Sigma$ implies $\Diamond\Box\psi\in\Sigma$,
but also that $\Box\psi\in\Sigma$ implies $\Diamond\Box\psi\in\Sigma$.
Thereby $\Sigma$ becomes infinite again.

\citep[p. 176]{Popkorn1994} uses an infinite set for filtration to
prove the finite model property of $\mathbf{K5}$. The author first
defines $\Box\Gamma=\{\Box\varphi\mid\varphi\in\Gamma\}$ and $\Diamond\Gamma=\{\Diamond\varphi\mid\varphi\in\Gamma\}$,
where $\Gamma=Sub\varphi$. Then he defines $\Gamma^{\bullet}=\Gamma\cup\Box\Gamma\cup\Diamond\Gamma$
and $\Gamma_{n}$ below.
\[
\begin{aligned}\Gamma_{0} & =\Gamma\\
\Gamma_{n+1} & =(\Gamma_{n})^{\bullet}
\end{aligned}
\]
Finally he defines $\Gamma^{*}=\bigcup_{n\geq0}\Gamma_{n}$, which
is just the smallest set that contains $Sub\varphi$ and is closed
under prefixing $\Box$ and $\Diamond$ . Then $\Gamma^{*}$ is used
for filtration. Though $\Gamma^{*}$ is an infinite set, it has finite
base with respect to Euclidean models. The proof is almost right,
except that $\Gamma^{*}$ is not subformula closed. If both $\Box$
and $\Diamond$ are primitive, then $\Gamma^{*}$ is indeed subformulas
closed. But the textbook takes only $\Box$ to be primitive. Then
$\Diamond\varphi$ is just the abbreviation of $\neg\Box\neg\varphi$,
whose subformulas include not only $\varphi$ but also $\neg\varphi$.
By the definition of $\Gamma^{*}$, we can not ensure that $\neg\varphi$
is also in $\Gamma^{*}$ if $\varphi\in\Gamma^{*}$. This is a very
elusive mistake. Compared to the mistake in \citep[p. 145]{Chagrov1997}
above, this mistake is more easily to be corrected: just take $Sub\Gamma^{*}$
instead of $\Gamma^{*}$ for filtration. Another option is to take
both $\Box$ and $\Diamond$ to be primitive in the language. Then
$\Gamma^{*}$ is subformula closed. Though in verifying that the filtration
model is Euclidean, we have to consider both $(2)$ and $(2')$ in
Definition~\ref{def:minimal-maximal-filtration}, the proof still
goes through, since $\Gamma^{*}$ is closed under prefixing both $\Box$
and $\Diamond$. 

A better proof is to take the hint in Exercise 2.3.8 on Page 83 of
\citep{Blackburn2001a}: using the smallest subformula closed set
that includes $\varphi$ and is closed under prefixing $\Box$ for
filtration (assuming $\mathcal{L}_{\Diamond}$ is our formal language).
We give the complete proof as follows for reference of teaching.
\begin{defn}
\label{def:finite-base}Given a set of formulas $\Sigma\subseteq\mathcal{L}_{\Diamond}$
and a model $\mathfrak{M}=(W,R,V)$, $\Sigma$ has finite base with
respect to $\mathfrak{M}$, if there exists a finite set $\Delta\subseteq\mathcal{L}_{\Diamond}$
such that for every $\varphi\in\Sigma$ there exists $\psi\in\Delta$
such that for all $w\in W$, $\mathfrak{M},w\models\varphi$ iff $\mathfrak{M},w\models\psi$.
\end{defn}

The following two propositions are easily verified.
\begin{prop}
\label{prop:finite-base}If $\Sigma$ has finite base with respect
to $\mathfrak{M}$, then all filtrations of $\mathfrak{M}$ through
$\Sigma$ is a finite model.
\end{prop}

\begin{prop}
\label{prop:Euclidean-reduction}The following formulas are valid
in Euclidean frames. 
\begin{enumerate}
\item $\Diamond\Diamond\Diamond\varphi\leftrightarrow\Diamond\Diamond\varphi$,
$\Box\Box\Box\varphi\leftrightarrow\Box\Box\varphi$
\item $\Diamond\Diamond\Box\varphi\leftrightarrow\Diamond\Box\varphi$,
$\Box\Box\Diamond\varphi\leftrightarrow\Box\Diamond\varphi$
\item $\Diamond\Box\Diamond\varphi\leftrightarrow\Diamond\Diamond\varphi$,
$\Box\Diamond\Box\varphi\leftrightarrow\Box\Box\varphi$
\item $\Diamond\Box\Box\varphi\leftrightarrow\Diamond\Box\varphi$, $\Box\Diamond\Diamond\varphi\leftrightarrow\Box\Diamond\varphi$
\end{enumerate}
\end{prop}

\begin{prop}
\label{prop:K5}$\mathbf{K5}$ has finite model property.
\end{prop}

\begin{proof}
It suffices to prove that a filtration of any Euclidean model is also
Euclidean. Given $\varphi\notin\mathbf{K5}$, let $\Sigma$ be the
the smallest subformula closed set that includes $\varphi$ and is
closed under prefixing $\Box$. By Proposition~\ref{prop:Euclidean-reduction},
$\Sigma$ has finite base with respect to Euclidean models. Suppose
$\mathfrak{M}=(W,R,V)$ is Euclidean. Consider the largest filtration
$\mathfrak{M}^{f}=(W^{f},R^{l},V^{f})$ of $\mathfrak{M}$ through
$\Sigma$. We prove that $\mathfrak{M}^{f}$ is also Euclidean. Suppose
$|w|R^{l}|u|$ and $|w|R^{l}|v|$. We prove $|u|R^{l}|v|$. Suppose
$\Diamond\psi\in\Sigma$ and $\mathfrak{M},v\models\psi$. It suffices
to prove $\mathfrak{M},u\models\Diamond\psi$. By $\mathfrak{M},v\models\psi$
and $|w|R^{l}|v|$, we have $\mathfrak{M},w\models\Diamond\psi$.
Since $\mathfrak{M}$ is Euclidean, it follows that $\mathfrak{M},w\models\Box\Diamond\psi$.
By the construction of $\Sigma$, we have $\Box\Diamond\psi\in\Sigma$.
Then by Remark~\ref{rem:filtration} and $|w|R^{l}|u|$, we obtain
$\mathfrak{M},u\models\Diamond\psi$.
\end{proof}

\section{Minimal Canonical Model and Finite Model Property}

Given any normal modal logic $\Lambda$, if $\varphi\notin\Lambda$
then it is easily seen that the canonical model of $\Lambda$ is a
counter-model of $\varphi$. But standard canonical models are not
finite models. By constructing canonical models using maximal consistent
sets relative to a finite set of formulas, we can prove the finite
model property of some logics. Such canonical models are called minimal
canonical models. The accessibility relations of minimal canonical
models for different logics are usually defined differently. In the
sequel, we will show an easily made mistake in defining minimal canonical
models, by taking the logic $\mathbf{KL}$ as an example.

Given any $\varphi$, let $Sub^{-}\varphi=\{\neg\psi\mid\psi\in Sub\varphi\}$.
Let $Sub^{+}\varphi=Sub\varphi\cup Sub^{-}\varphi$. Given a logic
$\Lambda$ and a formula $\varphi$, $\Gamma$ is $\Lambda$-maximal-consistent
relative to $\varphi$, if $\Gamma\subseteq Sub^{+}\varphi$ is $\Lambda$-consistent,
and for all $\psi\in Sub\varphi$, $\psi\in\Gamma$ or $\neg\psi\in\Gamma$.
\begin{defn}
\label{def:minimal-canonical}Given a formula $\varphi$, define the
minimal canonical model $\mathfrak{M}_{\varphi}^{\mathbf{KL}}=(W_{\varphi}^{\mathbf{KL}},R_{\varphi}^{\mathbf{KL}},V_{\varphi}^{\mathbf{KL}})$
relative to $\varphi$ for $\mathbf{KL}$ as follows.
\begin{enumerate}
\item $W_{\varphi}^{\mathbf{KL}}=\{\Gamma\subseteq Sub^{+}\varphi\mid\Gamma\text{ is }\Lambda\text{-maximal-consistent relative to }\varphi\}$.
\item For all $w,u\in W_{\varphi}^{\mathbf{KL}}$, $wR_{\varphi}^{\mathbf{KL}}u$
iff for all $\psi\in\mathcal{L}_{\Diamond}$, if $\Diamond\psi\in u$
or $\psi\in u$ then $\Diamond\psi\in w$, and there exists $\Diamond\chi\in w$
but $\Diamond\chi\notin u$.
\item For all $w\in W_{\varphi}^{\mathbf{KL}}$, if $p\in PV\cap Sub^{+}\varphi$,
then $w\in V_{\varphi}^{\mathbf{KL}}(p)$ iff $p\in w$; if $p\in PV-Sub^{+}\varphi$
then $V_{\varphi}^{\mathbf{KL}}(p)$ is an arbitrary subset of $W_{\varphi}^{\mathbf{KL}}$
(for instance, the empty set).
\end{enumerate}
\end{defn}

\begin{rem}
Unlike canonical models, in defining accessibility relations of minimal
canonical models, we can not freely choose $\Box$ or $\Diamond$,
but have to choose it according to which is primitive in the formal
language. Definition~\ref{def:minimal-canonical} only applies to
the case that $\Diamond$ is primitive. If the primitive modal operator
is $\Box$, to prove the truth lemma of minimal canonical models,
$(2)$ above should be replaced by $(2')$.
\begin{enumerate}
\item [$(2')$] For all $w,u\in W_{\varphi}^{\mathbf{KL}}$, $wR_{\varphi}^{\mathbf{KL}}u$
iff for all $\psi$, if $\Box\psi\in w$ then $\psi,\Box\psi\in u$
and there exists $\Box\chi\in u$ but $\Box\chi\notin w$.
\end{enumerate}
Neither $(2)$ nor $(2')$ implies the other (and hence they are not
equivalent). The reason is that, $w$ and $u$ here are not maximal
consistent sets, but maximal consistent sets relative to $\varphi$.
They have maximality only for the subformulas of $\varphi$, not for
all formulas. Thus they do not have some properties of maximal consistent
sets. For instance, from $\psi\in w$ and the logical equivalence
of $\psi$ and $\psi'$ it does not follow that $\psi'\in w$, because
$\psi'$ may not be in $Sub^{+}\varphi$. In particular, from $\Diamond\psi\notin w$
we can only obtain $\neg\Diamond\psi\in w$ but not $\Box\neg\psi\in w$.
\end{rem}

Just at this place, \citep{Blackburn2001a} made an incautious mistake.
Exercise 4.8.7 on Page 246 of this book is to prove the finite model
property of $\mathbf{KL}$. It gives a hint for proof: define the
canonical relation $R_{\varphi}^{\mathbf{KL}}$ by: $wR_{\varphi}^{\mathbf{KL}}u$
iff for all $\psi$, if $\Box\psi\in w$ then $\psi,\Box\psi\in u$,
and there exists $\Box\chi\in u$ but $\Box\chi\notin w$. This definition
is just $(2')$ above instead of (2). But \citep{Blackburn2001a}
takes $\Diamond$ to be primitive rather than $\Box$. This implies
that the correct hint should be $(2)$ rather than $(2')$. We guess
that the hint was copied from other sources which take $\Box$ to
be primitive and was forgotten to be revised accordingly. This might
be a careless mistake. But it could also be that the authors thought
$(2)$ and $(2')$ are just equivalent. If the latter is the case,
then our correction could be of more significance.

Of course, we could take both $\Box$ and $\Diamond$ to be primitive
in the formal language. Then subformulas are clearer, which will cause
less mistakes. But then in defining filtration and (minimal) canonical
models, we have to consider both $\Box$ and $\Diamond$. Moreover,
in proving some critical results (such as the truth lemma of canonical
models), we have to consider both the cases $\varphi=\Box\psi$ and
$\varphi=\Diamond\psi$. Though it does not increase difficulty, the
proof will become more tedious. From the perspective of writing textbooks,
taking only one modal operator as primitive is definitely more convenient
(but requires more care).

\section{Deductive Consequence and the Deduction Theorem}

Another easily made mistake is regarding the rule of necessitation
$\mathrm{RN}$ as having the same effect as the rule of modus ponens.
Though both preserve validity, they are different in preserving truth.
Modus ponens preserves local truth, i.e., for every world $w$, $\psi$
is truth at $w$ as long as $\varphi$ and $\varphi\to\psi$ are true
at $w$. Necessitation does not preserve local truth: from the truth
of $\varphi$ at $w$ it does not follow that $\Box\varphi$ is true
at $w$. Necessitation only preserves global truth, i.e., for every
model $\mathfrak{M}$, if $\varphi$ is true at all wolds in $\mathfrak{M}$,
then $\varphi$ is true at all worlds in $\mathfrak{M}$.\footnote{See \citep{Avron1991b,Fagin1992} and \citep{Sundholm1983} for the
distinction of these two kinds of rules.}

As the rule of necessitation does not preserve local truth, it can
not be applied to inferences with premises. This implies that when
defining the deductive consequence $\vdash_{\Lambda}$ for a modal
logic $\Lambda$, we can not follow the definition used for classical
propositional logic, which is defined by: $\Gamma\vdash_{\Lambda}\varphi$
iff there exists a finite sequence of formulas $\varphi_{1},\ldots,\varphi_{n}$
such that $\varphi_{n}=\varphi$ and for each $\varphi_{i}$, either
$\varphi_{i}\in\Gamma$, or $\varphi_{i}$ is an axiom of $\Lambda$,
or $\varphi_{i}$ is obtained from proceeding formulas in the sequence
using inference rules of $\Lambda$. If this definition is adopted
for modal logic, then as long as $\Lambda$ contains the rule of necessitation,
we have $p\vdash_{\Lambda}\Box p$. But the inference from $p$ to
$\Box p$ is usually invalid: from it is raining today, we can not
infer that it is necessary that it is raining today.

The understanding of the rule of necessitation is also related to
the deduction theorem in modal logic. Indeed, whether the deduction
theorem holds for modal logic had caused debate in the literature.\citep{Hakli2012}
To get the rule of necessitation right, there are six ways of defining
deductive consequence in modal logic in the literature.\footnote{In the sequel, we only consider deductive consequence defined in axiomatic
systems, and will not consider other proof systems (like sequent calculus
and tableau systems).}

The first way is called \emph{omitted definition} (as in \citep{Hughes1996}).
This definition does not consider inferences with premises, and considers
only weak soundness and weak completeness. Then the rule of necessitation
can only be applied to axioms and theorems derived from axioms, and
will thus avoid the incorrect application of it. The practice which
considers only theorems and not consequences was reasonable early
when logicians cared only about the relation between logic and mathematics.
But with the development of more and more non-classical logics, theorems
of an axiomatic system can not completely characterize a logic any
more.\footnote{For example, Kleene's three valued logic does not have any theorems,
but it has valid inferences.} The practice which considers only theorems and not consequences is
gradually abandoned. A bit surprisingly, a relatively new and popular
textbook \citep{VanBenthem2010} also adopts omitted definition.

The second way is called \emph{classical definition} (as in \citep{Chagrov1997}
and \citep{Fagin1995a}). This definition follows the standard definition
of deductive consequence for classical propositional logic and transfers
the treatment of the rule of necessitation to the deduction theorem,
which either has different contents (as in \citep{Chagrov1997}) or
is augmented with additional constraints. The latter is also the treatment
for the rule of universal generalization in some early textbooks (\citep{Hamilton1978})
in mathematical logic. There are two drawbacks of classical definition.
First, the deduction theorem is made too complicated, which is not
friendly to students. Second, there is no strong soundness under this
definition, which means the deductive consequence does not characterize
valid inferences.

The third way is called \emph{ternary definition} (as in \citep{Fitting1998}).
In this definition, premises are divided into two parts: a global
part and a local part. The notation $\Gamma\vdash_{\Lambda}\Delta\Rightarrow\varphi$
means that $\varphi$ can be deduced from global premises $\Gamma$
and local premises $\Delta$.\footnote{Fitting's original notation is $\Gamma\vdash_{\Lambda}\Delta\longrightarrow\varphi$.
To distinguish $\longrightarrow$ from material implication in the
object language, we use $\Rightarrow$ instead.} The rule of necessitation can be applied to global premises but not
to local premises. Thus, the deduction becomes a ternary relation.
Accordingly, the deduction theorem splits into two for the two kinds
of premises. As this definition is too distinctive, it has not been
widely adopted.

The fourth way is called \emph{reduced definition} (as in \citep{Lemmon1977}
, \citep{Chellas1980}, \citep{Blackburn2001a}, \citep{Goldblatt1993},
and open textbook \citep{Zach2018}). In this definition, $\Gamma\vdash_{\Lambda}\varphi$
iff there exists a finite subset $\Phi$ of $\Gamma$ such that $\bigwedge\Phi\to\varphi$\footnote{If $\Phi=\{\varphi_{1},\ldots,\varphi_{n}\}$ then $\bigwedge\Phi=\varphi_{1}\land\ldots\land\varphi_{n}$;
if $\Phi=\emptyset$ then $\bigwedge\Phi=\top$.} has a formal proof in $\Lambda$. The essential idea is to reduce
deduction with premises to that without premises. Since $\vdash_{\Lambda}p\to\Box p$
usually does not hold, under this definition $p\vdash_{\Lambda}\Box p$
does not hold either, and hence the deductive consequence accords
to the semantics. Moreover, the deduction theorem holds almost trivially
without any constraints. The rule of uniform substitution can also
be explicit without hidden in axiom schemata. Under this definition,
system $\mathbf{K}$ and its deductive consequence is defined as follows.
\begin{defn}
[reduced definition]The axiomatic system $\mathbf{K}^{r}$ consists
of the following axioms and rules of inference.
\begin{lyxlist}{00.00.0000}
\item [{PC1}] $p\to(q\to p)$
\item [{PC2}] $(p\to(q\to r)\to((p\to q)\to(p\to r))$
\item [{PC3}] $(\neg p\to\neg q)\to(q\to p)$
\item [{K}] $\Box(p\to q)\to(\Box p\to\Box q)$
\item [{dual}] $\Diamond p\leftrightarrow\neg\Box\neg p$
\item [{MP}] $\dfrac{\varphi,\varphi\to\psi}{\psi}$
\item [{RN}] $\dfrac{\varphi}{\Box\varphi}$
\item [{US}] $\dfrac{\varphi}{\varphi^{\sigma}}$, where $\sigma$ is any
substitution.\footnote{For any $\sigma:PV\to\mathcal{L}_{\Diamond}$, $\varphi^{\sigma}$
is inductively defined as follows. $p^{\sigma}=\sigma(p)$, $(\neg\psi)^{\sigma}=\neg\psi^{\sigma}$,
$(\psi\to\chi)^{\sigma}=\psi^{\sigma}\to\chi^{\sigma}$, $(\Diamond\psi)^{\sigma}=\Diamond\psi^{\sigma}$.}
\end{lyxlist}
We say $\varphi$ is derivable from $\Gamma$ in $\mathbf{K}^{r}$,
denoted $\Gamma\vdash_{\mathbf{K}^{r}}\varphi$, if there exists a
finite subset $\Phi$ of $\Gamma$ such that $\bigwedge\Phi\to\varphi$
has a formal proof in $\mathbf{K}^{r}$, i.e., there exists a finite
sequence of formulas $\varphi_{1},\ldots,\varphi_{n}$ such that $\varphi_{n}=\varphi$,
and for each $\varphi_{i}$, either $\varphi_{i}$ is an axiom of
$\mathbf{K}^{r}$, or is obtained from proceeding formulas in the
sequence using the rules of inference of $\mathbf{K}^{r}$. We abbreviate
$\emptyset\vdash_{\mathbf{K}^{r}}\varphi$ as $\vdash_{\mathbf{K}^{r}}\varphi$.
\end{defn}

Reduced definition is the most popular one in modal logic for deductive
consequence. But it has three drawbacks. First, it is not as intuitive
as classical definition, especially for students who learned classical
definition in classical propositional logic before. Without explanation,
they may be confused why in modal logic the definition has to be modified.
Second, technically the definition relies on the logical constant
$\to$\footnote{The other logical constant $\land$ can be eliminated by defining
$\Gamma\vdash_{\Lambda}\varphi$ by: there exists $\{\varphi_{1},\ldots,\varphi_{n}\}\subseteq\Gamma$
such that $\varphi_{1}\to(\varphi_{2}\to\cdots\to(\varphi_{n}\to\varphi)\cdots)$
has a formal proof in $\Lambda$.}. But not all logical languages contain this logical constant. Thereby,
reduced definition is not general enough. Third, though the deduction
theorem holds under this definition, it has no practical use at all.
To prove $\varphi\to\psi$, we can not assume $\varphi$ and prove
$\psi$ (so the proof may go easier) but have to prove $\varphi\to\psi$
directly, since $\varphi\vdash_{\Lambda}\psi$ is just defined by
$\vdash_{\Lambda}$$\varphi\to\psi$.

The fifth way is called \emph{bounded definition} (as in \citep{Hakli2012}).
This definition restricts the use of the rule of necessitation, so
that it can only be applied to conclusions derived from empty premises.
Under this definition, the deduction theorem also holds without constraints.
Compared to the above definitions, it has more advantages. But it
can not be easily defined precisely in standard forms of axiomatic
systems.\footnote{The definition in \citep{Hakli2012} is not given in standard forms
of axiomatic systems but in a form more like sequent calculus.}Here we give a precise definition that is closer to standard axiomatic
systems, taking system $\mathbf{K}$ as an example.
\begin{defn}
[bounded definition]\label{def:derivation-constrained}The axiomatic
system $\mathbf{K}^{b}$ consists of the following axiom schemata
and rules of inference.
\begin{lyxlist}{00.00.0000}
\item [{sPC1}] $\varphi\to(\psi\to\varphi)$
\item [{sPC2}] $(\varphi\to(\psi\to\chi)\to((\varphi\to\psi)\to(\varphi\to\chi))$
\item [{sPC3}] $(\neg\varphi\to\neg\psi)\to(\psi\to\varphi)$
\item [{sK}] $\Box(\varphi\to\psi)\to(\Box\varphi\to\Box\psi)$
\item [{sdual}] $\Diamond\varphi\to\neg\Box\neg\varphi$
\item [{MP}] $\dfrac{\varphi,\varphi\to\psi}{\psi}$
\item [{RN}] $\dfrac{\varphi}{\Box\varphi}$
\end{lyxlist}
A finite sequence of formulas $\varphi_{1},\ldots,\varphi_{n}$ is
a formal proof of $\varphi$ in $\mathbf{K}^{b}$ from empty premises,
if $\varphi_{n}=\varphi$, and for each $\varphi_{i}$
\begin{itemize}
\item either $\varphi_{i}$ is an instance of one of the axiom schema of
$\mathbf{K}^{b}$, or
\item there exists $j,k<i$ such that $\varphi_{k}=\varphi_{j}\to\varphi_{i}$,
i.e., $\varphi_{i}$ can be obtained from proceeding formulas in the
sequence using the rule MP, or
\item there exists $j<i$ such that $\varphi_{i}=\Box\varphi_{j}$, i.e.,
$\varphi_{i}$ can be obtained from a proceeding formula in the sequence
using RN.
\end{itemize}
We say $\varphi$ is derivable from $\Gamma$ in $\mathbf{K}^{b}$,
denoted $\Gamma\vdash_{\mathbf{K}^{b}}\varphi$, if there exists a
finite sequence of formulas $\varphi_{1},\ldots,\varphi_{n}$ such
that $\varphi_{n}=\varphi$ and for each $\varphi_{i}$,
\begin{itemize}
\item either $\varphi_{i}\in\Gamma$, or 
\item $\varphi_{i}$ is an instance of one of the axiom schema of $\mathbf{K}^{b}$,
or
\item there exists $j,k<i$ such that $\varphi_{k}=\varphi_{j}\to\varphi_{i}$,
i.e., $\varphi_{i}$ can be obtained from proceeding formulas in the
sequence using the rule MP, or
\item there exists $j<i$ such that $\varphi_{i}=\Box\varphi_{j}$, and
there exists $\varphi_{j_{1}},\ldots,\varphi_{j_{k}}\in\{\varphi_{1},\ldots,\varphi_{j}\}$
such that $\varphi_{j_{1}},\ldots,\varphi_{j_{k}}$ is a formal proof
of $\varphi_{j}$ in $\mathbf{K}^{b}$ from empty premises.
\end{itemize}
We abbreviate $\emptyset\vdash_{\mathbf{K}^{b}}\varphi$ by $\vdash_{\mathbf{K}^{b}}\varphi$. 
\end{defn}

We list some easily proved results for $\vdash_{\mathbf{K}^{b}}$
without proofs.
\begin{lem}
\label{lem:varphi-varphi}$\vdash_{\mathbf{K}^{b}}\varphi\to\varphi$.
\end{lem}

\begin{thm}
[Deduction Theorem]If $\Gamma,\psi\vdash_{\mathbf{K}^{b}}\varphi$,
then $\Gamma\vdash_{\mathbf{K}^{b}}\psi\to\varphi$.
\end{thm}

\begin{thm}
[Compactness]\label{thm:K^b-compactness}$\Gamma\vdash_{\mathbf{K}^{b}}\varphi$
iff there exists a finite subset $\Phi$ of $\Gamma$ such that $\Phi\vdash_{\mathbf{K}^{b}}\varphi$.
\end{thm}

\begin{cor}
[Derivation Theorem]\label{cor:K^b-derivation}$\Gamma\vdash_{\mathbf{K}^{b}}\varphi$
iff there exists a finite subset $\Phi$ of $\Gamma$ such that $\vdash_{\mathbf{K}^{b}}\bigwedge\Phi\to\varphi$.
\end{cor}

The derivation theorem indicates that the deduction relation $\vdash_{\mathbf{K}^{b}}$
can be characterized by the theorems of $\mathbf{K}^{b}$. The next
theorem is used for proving Theorem~\ref{thm:equivalence}.
\begin{thm}
[Substitution Theorem]\label{thm:K^b-substitution}If$\vdash_{\mathbf{K}^{b}}\varphi$
then $\vdash_{\mathbf{K}^{b}}\varphi^{\sigma}$, for any substitution
$\sigma$.
\end{thm}

The sixth way is called \emph{deflationary definition}. This definition
hides the rule of necessitation in axiom schemata, so that the set
of axioms is closed under prefixing $\Box$. Then the rule of necessitation
is redundant. According to \citep{Hakli2012}, this definition was
suggested in \citep{Smorynski1984} for provability logic. The idea
may come from some textbooks in mathematical logic (as \citep{Enderton2001})
, in which the rule of universal generalization is hidden in axiom
schemata so that the set of axioms is closed under taking universal
quantification. The advantage of this definition is that we can follow
the standard definition of deductive consequence in classical propositional
logic without destroying the deduction theorem, and strong soundness
is also maintained. Unfortunately, no textbooks in modal logic adopts
this definition. That said, this definition also has a drawback. For
logics that does not contain $\mathrm{RN}$ but only weaker rules
like $\mathrm{RE}$, this definition is not applicable any more.

In the sequel, we take $\mathbf{K}$ as an example to give the precise
deflationary definition and show its equivalence to reduced definition
and bounded definition for reference of teaching.
\begin{defn}
[deflationary definition]\label{def:derivation-elimination}The
set of axioms $\Xi_{\mathbf{K}^{d}}$ of the axiomatic system $\mathbf{K}^{d}$
includes
\begin{lyxlist}{00.00.0000}
\item [{(PC1)}] $p\to(q\to p)$
\item [{(PC2)}] $(p\to(q\to r)\to((p\to q)\to(p\to r))$
\item [{(PC3)}] $(\neg p\to\neg q)\to(q\to p)$
\item [{(K)}] $\Box(p\to q)\to(\Box p\to\Box q)$
\item [{(dual)}] $\Diamond p\leftrightarrow\neg\Box\neg p$, and satisfies
the following conditions.
\end{lyxlist}
\begin{itemize}
\item It is closed under necessitation, i.e., if $\varphi\in\Xi_{\mathbf{K}^{d}}$
then $\Box\varphi\in\Xi_{\mathbf{K}^{d}}$.
\item It is closed under substitution, i.e., if $\varphi\in\Xi_{\mathbf{K}^{d}}$
then $\varphi^{\sigma}\in\Xi_{\mathbf{K}^{d}}$, where $\sigma$ is
any substitution.
\item No other formulas are in $\Xi_{\mathbf{K}^{d}}$.
\end{itemize}
The only rule of inference of $\mathbf{K}^{d}$ is 
\begin{lyxlist}{00.00.0000}
\item [{(MP)}] $\dfrac{\varphi,\varphi\to\psi}{\psi}$.
\end{lyxlist}
A finite sequence of formulas \emph{$\varphi_{1},\ldots,\varphi_{n}$}
is a formal proof of $\varphi$ from $\Gamma$ in $\mathbf{K}^{d}$\emph{,
if }$\varphi=\varphi_{n}$ and for each $\varphi_{i}$,
\begin{itemize}
\item either $\varphi_{i}\in\Gamma$, or
\item $\varphi_{i}\in\Xi_{\mathbf{K}^{d}}$, or 
\item there exist $j,k<i$ such that $\varphi_{k}=\varphi_{j}\to\varphi_{i}$,
i.e., $\varphi_{i}$ is obtained from proceeding formulas in the sequence
using modus ponens.
\end{itemize}
We say $\varphi$ is derivable from $\Gamma$ in $\mathbf{K}^{d}$,
denoted $\Gamma\vdash_{\mathbf{K}^{d}}\varphi$, if there is a formal
proof of $\varphi$ from $\Gamma$ in $\mathbf{K}^{d}.$ We abbreviate
$\emptyset\vdash_{\mathbf{K}^{d}}\varphi$ as $\vdash_{\mathbf{K}^{d}}\varphi$.

\end{defn}

\begin{thm}
[Deduction Theorem]If $\Gamma,\psi\vdash_{\mathbf{K}^{d}}\varphi$,
then $\Gamma\vdash_{\mathbf{K}^{d}}\psi\to\varphi$.
\end{thm}

\begin{thm}
[Compactness]$\Gamma\vdash_{\mathbf{K}^{d}}\varphi$ iff there exists
a finite subset $\Phi$ of $\Gamma$ such that $\Phi\vdash_{\mathbf{K}^{d}}\varphi$.
\end{thm}

\begin{cor}
[Derivation Theorem]\label{cor:K^e-derivation}$\Gamma\vdash_{\mathbf{K}^{d}}\varphi$
iff there exists a finite subset $\Phi$ of $\Gamma$ such that $\vdash_{\mathbf{K}^{d}}\bigwedge\Phi\to\varphi$.
\end{cor}

The derivation theorem indicates that the deduction relation $\vdash_{\mathbf{K}^{d}}$
can be characterized by the theorems of $\mathbf{K}^{d}$.
\begin{thm}
[Generalization Theorem]If $\Gamma\vdash_{\mathbf{K}^{d}}\varphi$
then $\Box\Gamma\vdash_{\mathbf{K}^{d}}\Box\varphi$.
\end{thm}

The following results is easily obtained from the generalization theorem.
\begin{cor}
\label{cor:K^e-admissible-rule}$\mathbf{K}^{d}$ has the following
admissible rules.\footnote{\label{fn:admissible-rule}A rule $\dfrac{\Gamma}{\varphi}$ is admissible
in $\mathbf{S}$, if for any substitution $\sigma$, $\vdash_{\mathbf{S}}\Gamma^{\sigma}$
implies $\vdash_{\mathbf{S}}\varphi^{\sigma}$, where $\vdash_{\mathbf{S}}\Gamma^{\sigma}$
means $\vdash_{\mathbf{S}}\psi^{\sigma}$ for all $\psi\in\Gamma$.}
\begin{enumerate}
\item $\dfrac{\varphi}{\Box\varphi}$
\item $\dfrac{\varphi\to\psi}{\Box\varphi\to\Box\psi}$
\item $\dfrac{\varphi\leftrightarrow\psi}{\Box\varphi\leftrightarrow\Box\psi}$
\item $\dfrac{\varphi_{1}\land\ldots\land\varphi_{n}\to\varphi}{\Box\varphi_{1}\land\ldots\land\Box\varphi_{n}\to\Box\varphi}$
\end{enumerate}
\end{cor}

\begin{thm}
\label{thm:equivalence}$\vdash_{\mathbf{K}^{r}}\ =\ \vdash_{\mathbf{K}^{b}}\ =\ \vdash_{\mathbf{K}^{d}}$
\end{thm}

\begin{proof}
We prove $\vdash_{\mathbf{K}^{r}}\ \subseteq\ \vdash_{\mathbf{K}^{b}}\ \subseteq\ \vdash_{\mathbf{K}^{e}}\ \subseteq\ \vdash_{\mathbf{K}^{r}}$.

$\vdash_{\mathbf{K}^{r}}\ \subseteq\ \vdash_{\mathbf{K}^{b}}$: By
Corollary~\ref{cor:K^b-derivation}, it suffices to prove that for
all $\varphi\in\mathcal{L}_{\Diamond}$, if $\vdash_{\mathbf{K}^{r}}\varphi$
then $\vdash_{\mathbf{K}^{b}}\varphi$. By induction on the length
of the formal proof of $\varphi$ in $\mathbf{K}^{r}$. If $\varphi$
is an axiom of $\mathbf{K}^{r}$, then it is also an instance of an
axiom schema of $\mathbf{K}^{b}$. Thus $\vdash_{\mathbf{K}^{b}}\varphi$.
If $\varphi$ is obtained from $\psi$ and $\psi\to\varphi$ using
MP, then by the inductive hypothesis, $\vdash_{\mathbf{K}^{b}}\psi$
and $\vdash_{\mathbf{K}^{b}}\psi\to\varphi$. By MP of $\mathbf{K}^{b}$,
it follows that $\vdash_{\mathbf{K}^{b}}\varphi$. If $\varphi=\Box\psi$
is obtained from $\psi$ using RN, then by the inductive hypothesis,
$\vdash_{\mathbf{K}^{b}}\psi$. Then $\psi$ is derived from empty
premises. By the restricted RN of $\mathbf{K}^{b}$, $\vdash_{\mathbf{K}^{b}}\Box\psi$,
i.e., $\vdash_{\mathbf{K}^{b}}\varphi$. If $\varphi=\psi^{\sigma}$
is obtained from $\psi$ by US, by the inductive hypothesis, $\vdash_{\mathbf{K}^{b}}\psi$.
By Theorem~\ref{thm:K^b-substitution}, $\vdash_{\mathbf{K}^{b}}\psi^{\sigma}$,
i.e., $\vdash_{\mathbf{K}^{b}}\varphi$.

$\vdash_{\mathbf{K}^{b}}\ \subseteq\ \vdash_{\mathbf{K}^{d}}$: By
Corollary ~\ref{cor:K^b-derivation} and Corollary~\ref{cor:K^e-derivation},
it suffices to prove that for all $\varphi\in\mathcal{L}_{\Diamond}$,
if $\vdash_{\mathbf{K}^{b}}\varphi$ then $\vdash_{\mathbf{K}^{d}}\varphi$.
By induction on the length of the formal proof of $\varphi$ in $\mathbf{K}^{b}$.
If $\varphi$ is an instance of an axiom schema of $\mathbf{K}^{b}$,
then it is also an instance of an axiom schema of $\mathbf{K}^{d}$.
Thus $\vdash_{\mathbf{K}^{d}}\varphi$. If $\varphi$ is obtained
from $\psi$ and $\psi\to\varphi$ using MP, then by the inductive
hypothesis, $\vdash_{\mathbf{K}^{d}}\psi$ and $\vdash_{\mathbf{K}^{d}}\psi\to\varphi$.
By MP of $\mathbf{K}^{d}$, $\vdash_{\mathbf{K}^{d}}\varphi$. If
$\varphi=\Box\psi$ is obtained from $\psi$ using RN, then by the
inductive hypothesis, $\vdash_{\mathbf{K}^{d}}\psi$. By Corollary~\ref{cor:K^e-admissible-rule}(1),
$\vdash_{\mathbf{K}^{d}}\Box\psi$, i.e., $\vdash_{\mathbf{K}^{d}}\varphi$.

$\vdash_{\mathbf{K}^{d}}\ \subseteq\ \vdash_{\mathbf{K}^{r}}$: By
Corollary~\ref{cor:K^e-derivation}, it suffices to prove that for
all $\varphi\in\mathcal{L}_{\Diamond}$, if $\vdash_{\mathbf{K}^{d}}\varphi$
then $\vdash_{\mathbf{K}^{r}}\varphi$. By induction on the length
of the formal proof of $\varphi$ in $\mathbf{K}^{d}$. If $\varphi\in\Xi_{\mathbf{K}^{d}}$,
then we prove by further induction on $\varphi$. If $\varphi\in\{\mathrm{PC1,PC2,PC3,K,dual}\}$,
then it is also an axiom of $\mathbf{K}^{r}$. Thus $\vdash_{\mathbf{K}^{r}}\varphi$.
If $\varphi=\Box\psi$ and $\psi\in\Xi_{\mathbf{K}^{d}}$, then by
the inductive hypothesis, $\vdash_{\mathbf{K}^{r}}\psi$. Then by
RN of $\mathbf{K}^{r}$, $\vdash_{\mathbf{K}^{r}}\Box\psi$, i.e.,
$\vdash_{\mathbf{K}^{r}}\varphi$. If $\varphi=\psi^{\sigma}$ and
$\psi\in\Xi_{\mathbf{K}^{d}},$then by the inductive hypothesis, $\vdash_{\mathbf{K}^{r}}\psi$.
Then by US of $\mathbf{K}^{r}$, $\vdash_{\mathbf{K}^{r}}\psi^{\sigma}$,
i.e., $\vdash_{\mathbf{K}^{r}}\varphi$. Now we finish the case of
$\varphi\in\Xi_{\mathbf{K}^{e}}$. If $\varphi$ is obtained from
$\psi$ and $\psi\to\varphi$ using MP, then by the inductive hypothesis,
$\vdash_{\mathbf{K}^{r}}\psi$ and $\vdash_{\mathbf{K}^{r}}\psi\to\varphi$.
By MP of $\mathbf{K}^{r}$, we have $\vdash_{\mathbf{K}^{r}}\varphi$.
\end{proof}
The theorem says that the reduced definition, bounded definition,
and deflationary definition are all equivalent, which may help students
to understand better the concept of deduction consequence in modal
logic. We can also adopt the most convenient one among them according
to the demand in applications or proofs in different situations.

\section{Conclusion}

We discuss four common mistakes in the teaching and textbooks of modal
logic. The first one is missing the axiom $\Diamond\varphi\leftrightarrow\neg\Box\neg\varphi$,
when choosing $\Diamond$ as the primitive modal operator, misunderstanding
that $\Box$ and $\Diamond$ are symmetric. The second one is forgetting
to make the set of formulas for filtration closed under subformulas,
when proving the finite model property through filtration, neglecting
that $\Box\varphi$ and $\Diamond\varphi$ may be abbreviations of
formulas. The third one is giving wrong definitions of canonical relations
in minimal canonical models that are unmatched with the primitive
modal operators. These three mistakes are all related to the choice
of primitive modal operators. The moral is that we can not take for
granted that $\Box$ and $\Diamond$ are symmetric and can be chosen
arbitrarily in certain definitions. They have to be carefully selected
according to which modal operator is taken as primitive.

The fourth mistake is misunderstanding the rule of necessitation,
without knowing its distinction from the rule of modus ponens. We
summarizes six methods of defining deductive consequence in modal
logic: omitted definition, classical definition, ternary definition,
reduced definition, bounded definition, and deflationary definition,
and show that the last three definitions are equivalent to each other.


\end{document}